\documentclass[aps,superscriptaddress,twocolumn,twoside,floatfix,prl,a4paper]{revtex4-1} 

\usepackage[normalem]{ulem}
 \usepackage{lipsum}
\usepackage{times}
\usepackage{epsfig}
\usepackage{amsfonts}
\usepackage{amsmath}
\usepackage{amssymb,amsthm}
\usepackage{xcolor,colortbl}
\usepackage{multirow}
\usepackage{braket}
\usepackage{latexsym}
\usepackage{amsfonts}
\usepackage{mathrsfs}
\usepackage{natbib}
\usepackage{verbatim}
\usepackage{gensymb}
\usepackage{caption}
\usepackage{mathrsfs}
\usepackage{caption}
\usepackage{subcaption}
\usepackage{ragged2e}
\DeclareCaptionJustification{justified}{\justifying}
\usepackage{blkarray}
\usepackage{graphicx}
\newtheorem{theorem}{Theorem}
\newtheorem{corollary}{Corollary}

\newtheorem{proposition}{Proposition}

\usepackage[colorlinks=true,linkcolor=blue,citecolor=blue,urlcolor=blue]{hyperref}
\allowdisplaybreaks

\newtheorem{remark}{Remark}

\newcommand{\G}{\mathcal{G}}

\begin{document}

\title{{Superadditivity of Zero-Error Capacity in Noisy Classical and Perfect Quantum Channel Pairs}}


\author{Ambuj}
\affiliation{Indian Institute of Technology, Jodhpur-342030, India}

\author{Anushko Chattopadhyay}
\affiliation{Indian Institute of Technology, Jodhpur-342030, India}

\author{Kunika Agarwal}
\affiliation{Department of Physics of Complex Systems, S. N. Bose National Center for Basic Sciences, Block JD, Sector III, Salt Lake, Kolkata 700106, India.}

\author{Rakesh Das}
\affiliation{Indian Institute of Technology, Jodhpur-342030, India}

\author{Amit Mukherjee}
\affiliation{Indian Institute of Technology, Jodhpur-342030, India}

\begin{abstract}
We demonstrate superadditivity of one-shot zero-error classical capacity in an asymmetric communication setting where a noisy classical channel is used in parallel with a perfect quantum channel. Each channel individually supports only a fixed number of perfectly distinguishable messages. Their joint use enables transmission of strictly more messages than permitted by the product of the individual capacities. We present explicit constructions achieving this enhancement and establish that replacing the perfect quantum channel with a perfect classical channel eliminates the effect. Finally, we identify a structural criterion on the noisy channel governing this effect and show that the quantum advantage is rooted in Kochen–Specker contextuality.
\end{abstract}

\maketitle


\textit{Introduction.-}
Can the combined performance of two imperfect resources exceed the sum of what each can achieve individually. The origin of this philosophy traces back to Aristotle \cite{AristotleMetaphysics}. In information theory, this question takes a more concrete form in the study of communication channels, where performance is quantified by channel capacity \cite{Shannon1948}. When two noisy channels are used together, their joint capacity is generally expected to be equal to their sum of individual capacities. Remarkably, quantum information theory violates this intuition. Certain pairs of noisy quantum channels exhibit superadditivity, whereby their combined capacity exceeds the sum of the capacities achievable in isolation. A quantum channel supports the transmission of quantum information, private information, or classical information, leading to distinct notions of capacity \cite{Schumacher1997,Holevo1998,Lloyd1997,Devetak2005,Shor2002}. Superadditivity has been rigorously established for quantum and private capacities \cite{Li2009,Smith2009,Smith2008,Smith2011}, and indications of a similar phenomenon for classical capacity have also been reported \cite{Hastings2009}. 

Zero-error information theory, introduced by Shannon \cite{Shannon1956}, addresses communication under the stringent requirement of absolute certainty, allowing no probability of error. Within this setup, interestingly even noisy classical channels can display superadditive behavior \cite{Shannon1956,Lovasz1979,Alon1998,Korner1998}. More recently,  quantum extension of zero-error information theory has emerged as an active area of research \cite{Cubitt2010,Duan2009,Cubitt2011a,Cubitt2012,Shirokov2015,Duan2016,Cubitt2011,Chen2010,Beigi2010,Chen2011,Duan2013,Wang2018,Prevedel2011,Yadavalli2022,Mir2023,Agarwal2024,Cubitt2012,Shirokov2015}. In this work, we present a previously unexplored instance of superadditivity phenomenon having a quantum origin. Specifically, we show that a noisy classical channel, when paired in parallel with a perfect quantum channel, exhibits superadditive capacity. In other words, the capacity of the joint noisy classical and perfect quantum channels supersedes sum of their individual capacities. We also establish that this enhancement is fundamentally quantum, \textit{i.e.}, replacing the perfect quantum channel with a perfect classical channel eliminates the superadditive advantage. Along the way, we derive general criteria for exhibiting such superadditive behavior, thereby expanding the conceptual landscape of zero-error communication theory and connecting it to the foundational notion of Kochen Specker contextuality \cite{Kochen1975,KSRMP}. 

\textit{Framework.-}
A discrete memoryless classical channel $N$ is specified by a finite input alphabet $\mathcal{X}$, a finite output alphabet $\mathcal{Y}$, and a collection of conditional probabilities $P(y|x)$, where $\sum_{y \in \mathcal{Y}} P(y|x) = 1$ for all $x \in \mathcal{X}$, and $y \in \mathcal{Y}$.
For a fixed input symbol $x \in \mathcal{X}$, only a subset of output symbols can occur with nonzero probability \cite{Shannon1948}.
We denote this support of input by
$\Gamma(x) := \{ y \in \mathcal{Y} \mid P(y|x) > 0 \}$,
Operationally, $\Gamma(x)$ represents all outputs that the receiver may observe when $x$ is transmitted in a single use of the channel.

Two distinct input symbols $x,x' \in \mathcal{X}$ are said to be \emph{confusable} in the single use of the channel if their supports in $\mathcal{Y}$ overlap, that is, if $\Gamma(x) \cap \Gamma(x') \neq \varnothing$ \cite{Shannon1956}.
In this case, there exists at least one output that could have arisen from either input, and therefore the receiver cannot infer with certainty which input was sent.
Conversely, two input symbols are perfectly distinguishable in a single channel use if their supports are disjoint.

A \emph{one-shot zero-error code} for the channel is a subset $\mathcal{C} \subseteq \mathcal{X}$ such that every pair of distinct input symbols in $\mathcal{C}$ is perfectly distinguishable by the receiver.
Equivalently, for all $x \neq x'$ in $\mathcal{C}$, we require $\Gamma(x) \cap \Gamma(x') = \varnothing$.
The size of such a code, $|\mathcal{C}|$ represents the number of classical messages that can be transmitted with absolute certainty by invoking single use of the channel.

The confusability relations among input symbols admit a natural graph-theoretic representation.
The \emph{confusability graph} $G(N)$ of the channel has vertex set $V(G)$ which coincides with input set $\mathcal{X}$.
Two distinct vertices are connected by an edge if and only if the corresponding input symbols are confusable and the set of all edges is denoted by $E(G)$.
By construction, a one-shot zero-error code $\mathcal{C}$ corresponds precisely to a set of vertices with no edges between them.
In graph-theoretic terms, a code $\mathcal{C}$ is an independent set of the graph $G(N)$.
The {one-shot zero-error classical capacity }$\mathfrak{C}_0(N)$ of the channel is therefore given by the independence number $\alpha(G(N))$, defined as the \textit{maximum} cardinality of an independent set of $G(N)$.
Although capacity is often defined as the logarithm of this quantity, throughout this work we identify the one-shot zero-error capacity with $\alpha(G(N))$ itself, corresponding directly to the maximum number of messages that can be transmitted with certainty in a single channel use. As a simple example, a perfect classical channel with $d$ input symbols has an edgeless confusability graph ($\overline{K}_d$) and hence $\mathfrak{C}_0(\overline{K}_d)=\alpha(\overline{K}_d)=d$.

When two classical channels $N$ and $N'$ are used in parallel, denoted as $N \otimes N'$, the receiver receives both outputs simultaneously. Two input pairs are confusable if the receiver cannot distinguish at least one component of the pair in the corresponding channel perfectly.
This joint confusability graph corresponds precisely to the {strong product} of the individual confusability graphs, denoted $G(N) \boxtimes G(N')$. Two distinct vertices $(u,i)$ and $(v,j)$ in $G(N) \boxtimes G(N')$ are adjacent if and only if at least one of the following holds:
(i) $u = v$ and $i$ is adjacent to $j$ in $G(N')$; 
(ii) $u$ is adjacent to $v$ in $G(N)$ and $i = j$; 
(iii) $u$ is adjacent to $v$ in $G(N)$ and $i$ is adjacent to $j$ in $G(N')$.
In this framework, a one-shot zero-error code for the parallel channel corresponds to an independent set in the strong product graph $G(N) \boxtimes G(N')$. Consequently, the joint one-shot zero-error capacity is given by
$\mathfrak{C}_0(N \otimes N') = \alpha\bigl(G(N) \boxtimes G(N')\bigr)$.

A noisy classical channel $N$ can equivalently be represented by its {channel hypergraph}, denoted as $H(N)$. A vertex set of $H(N)$ is the input alphabet $\mathcal{X}$ of the channel. The output symbol $y \in \mathcal{Y}$, is a hyperedge $h_y$ which consists of vertex subset $\{x\in\mathcal{X} \mid P(y|x)>0\}\subseteq\mathcal{X}$.
Two inputs are confusable if and only if they belong to a common hyperedge.
The confusability graph $G(N)$ is obtained from $H(N)$ by connecting each pair of vertices with an edge that appear together in at least one hyperedge.

We now restrict our attention to the parallel composition of two classical channels: one noisy channel $N_n$, and one perfect channel $N_p$.
Our goal is to characterize the one-shot zero-error capacity of the joint channel.
In this case, the resulting capacity is simply the product of the individual capacities \cite{Shannon1956}.
We include a brief proof of this multiplicativity for completeness. It establishes the classical benchmark from which a departure is later observed when a perfect quantum channel is used in place of a perfect classical one.

\begin{proposition}[\cite{Shannon1956}]
Let $N_n$ be a noisy classical channel with confusability graph $G_n$ and one-shot zero-error capacity $\mathfrak{C}_0(N_n)=\alpha(G_n)$, and let $N_p$ be a $d$-level perfect classical channel with confusability graph $G_p$; then the one-shot zero-error capacity of their parallel composition $N_n \otimes N_p$ is $\mathfrak{C}_0(N_n\otimes N_p)=\alpha(G_n \boxtimes G_p) = d\, \alpha(G_n)$.
\end{proposition}

\begin{proof}
The confusability graph of the perfect classical channel $N_p$ is the edgeless graph $G_p=\overline{K}_d$, and hence its capacity is $\alpha(G_p)=d$.
We now analyze the confusability graph of the parallel composition $N_n\otimes N_p$.

The joint channel's confusability graph is the strong product of $G_n$ and $G_p$, that is $G_n \boxtimes G_p$, and the vertex set is $V(G_n) \times V(G_p)$.
Since $G_p$ is edgeless, the edge set $E(G_p)$ is empty. Hence, adjacency in $G_n \boxtimes G_p$ reduces to the condition: $(u,i) \sim (v,j)$ if and only if  $i = j$ and  $u \sim v$ in $G_n$.
It follows that $G_n \boxtimes G_p$ decomposes into $d$ disconnected components, one for each $i \in V(G_p)$, and each component is isomorphic to $G_n$ through the map $(u,i) \mapsto u$. Therefore,
$G_n \boxtimes G_p = \bigsqcup_{i=1}^{d} G_n^i$, where each $G^i_n \cong G_n$.

An independent set in a disjoint union of graphs is obtained by taking independent sets in each component. Thus, the maximum size of an independent set in $G_n \boxtimes G_p$ is the sum of the independence numbers of its components:
$\alpha(G_n \boxtimes G_p) = \sum_{i=1}^{d} \alpha(G_n^i) = d\,\alpha(G_n)$.
This proves the claim.
\end{proof}

The above proposition establishes that supplementing a noisy classical channel $N_n$ with a perfect classical channel does not give rise to any superadditivity of the one-shot zero-error capacity. This naturally raises the question of whether this limitation persists when the assisting perfect channel is quantum rather than classical.

\textit{Main Results.-} A perfect quantum channel can transmit the same number of classical messages with zero-error as a perfect classical channel of the same level \cite{Schumacher1997,Holevo1998}. Despite this equivalence at the level of individual channel capacities, we show that replacing the perfect classical channel with a perfect quantum channel allows the one-shot zero-error capacity of the joint channel to strictly exceed the classical bound $ \alpha(G(N_n)) $ multiplied by the capacity of the perfect channel. This superadditive phenomenon is intrinsically quantum in nature as no such superadditivity arises when the assisting channel is classical. We provide a constructive proof of this superadditivity in what follows.

\begin{theorem}\label{scheme}
There exist a noisy classical channel $N_n$ and a perfect quantum channel $Q_p$ such that, in a single parallel use of $N_n$ and $Q_p$, the one-shot zero-error capacity strictly exceeds the product of the individual one-shot zero-error capacities of $N_n$ and $Q_p$.
\end{theorem}

\begin{proof}
We prove the claim by explicit construction.
Let $N_n$ be a discrete memoryless classical channel whose hypergraph $H(N_n)$ is shown in Fig.~\ref{fig:c18}. The input (vertex) set is $\mathcal{X}=\{v_1,\dots,v_{18}\}$, and the output set $\mathcal{Y}$ consists of hyperedges $h_y$. Each hyperedge $h_y$ contains four inputs, which form a maximum clique in the confusability graph $G(N_n)$ of $N_n$. The graph $G(N_n)$ coincides with the orthogonality graph of the $18$-vector construction introduced by Cabello \emph{et al.} \cite{CABELLO1996183}, which exhibits Kochen Specker contextuality. The independence number of $G(N_n)$ is $\alpha(G(N_n)) = 4$, and therefore the one-shot zero-error classical capacity of $N_n$ satisfies $\mathfrak{C}_0(N_n)=4$.

Now we augment $N_n$ by forward-assisting it with a perfect $4$-level quantum channel $Q_p$, whose one-shot zero-error classical capacity is also $4$. We describe a joint encoding–decoding scheme for the parallel use of $N_n$ and $Q_p$. Let the message set be $\mathcal{M}_m:=\{m_1,\dots,m_{18}\}$.

\textit{Encoding.} Fix a bijection $\mathscr{E}:\mathcal{M}_m\to V(\mathcal{G}_m)$. To transmit a message $m_k\in\mathcal{M}_m$, let $v_k=\mathscr{E}(m_k)$ be the corresponding classical label; the sender sends this classical label through the noisy classical channel $N_n$. Simultaneously, the sender transmits a quantum system encoded in the (normalized) state
$\ket{\psi_k}=a_k\ket{0}+b_k\ket{1}+c_k\ket{2}+d_k\ket{3}\in\mathbb{C}^4$,
through $Q_p$, which corresponds to the vertex $v_k=(a_k,b_k,c_k,d_k)$ (See Fig. \ref{fig:c18}).

\textit{Decoding.} Upon reception, the receiver obtains a hyperedge $h_y$ (the classical channel output) together with the quantum state $\ket{\psi_k}$ (the quantum channel output). However, the state is unknown to the receiver. By construction, for any hyperedge $h_y$ of $H(N_n)$ the set $\{\ket{\psi_j} \mid v_j \in h_y\}$ is an orthonormal basis of $\mathbb{C}^4$. Hence each hyperedge $h_y$ determines a projective measurement consisting of the rank-one projectors $\mathbb{P}^{(y)}_j = \ket{\psi_j}\!\bra{\psi_j}$ for $v_j\in h_y$, with $\sum_{v_j \in h_y} \mathbb{P}^{(y)}_j = I_4$. After observing $h_y$, the receiver performs the corresponding projective measurement on the received quantum system. Since $\ket{\psi_k}$ is one of the basis vectors associated with $h_y$, the measurement outcome identifies $\ket{\psi_k}$ with certainty and thus recovers the transmitted message $m_k$ without error.

Therefore the protocol transmits all $18$ messages perfectly. It follows that the one-shot zero-error capacity of the combined use of $N_n$ and $Q_p$ is at least $18$. This achievable rate strictly exceeds the product of the individual capacities, $4\times 4=16$, demonstrating superadditivity of the zero-error capacity.
\end{proof}

\begin{figure}[t]
\centering
\includegraphics[width=\linewidth]{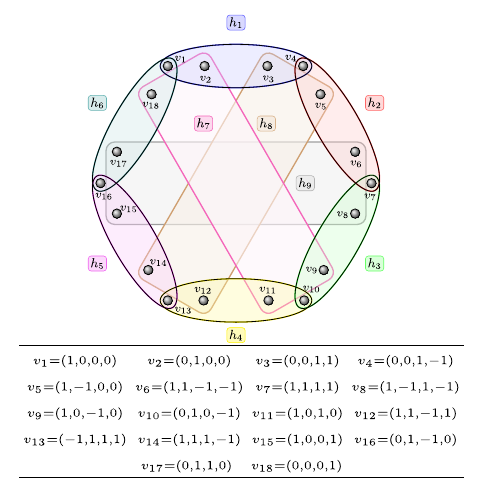}
\caption{Channel hypergraph associated with the noisy classical channel $N_n$ given in Theorem \ref{scheme}. The hypergraph has $18$ vertices, corresponding to the channel inputs, which are represented by vectors $v_i$ in $\mathbb{R}^4$ and listed in the table. There are 9 hyperedges corresponding to possible channel outputs. Each vertex is contained in exactly two hyperedges. Accordingly, for any given input, the noisy classical channel can produce either of the two associated outputs.}
\label{fig:c18}
\end{figure}

To place Theorem~\ref{scheme} in context, we compare it with prior work that also exploited quantum resources to enhance zero-error capacity. In particular, Cubitt \emph{et al.} \cite{Cubitt2010} demonstrated that entanglement assistance can result in a strict enhancement of the zero-error capacity for certain noisy classical channels.
Their analysis identifies a necessary structural condition for such an enhancement: the channel hypergraph must be Kochen Specker (KS) uncolorable, reflecting the impossibility of consistent deterministic classical value assignment.

The noisy classical channel we use to prove Theorem~\ref{scheme} has the channel hypergraph constructed from the $18$-vector Kochen Specker set of Cabello \emph{et al.} (Fig.~\ref{fig:c18}), so it is therefore KS uncolorable. Importantly, however, our protocol achieves superadditivity in cases beyond KS uncolorable hypergraphs, showing that the structural requirement identified for entanglement-assisted protocols does not carry over as a necessary condition for the effect we observe.

\begin{remark}\label{sic}
Consider a noisy classical channel $N_n$ whose hypergraph has as vertices the vectors in $\mathbb{C}^3$:
$(1,0,0)$, $(0,1,0)$, $(0,0,1)$, $(1,-\omega^i,0)$, $(1,0,-\omega^i)$, $(0,1,-\omega^i)$, $(1,\omega^i,\omega^j)$,
where $\omega = e^{2\pi i/3m}$ and $i,j\in\{1,\dots,3m\}$. For $m\in\mathbb Z_+$, this family of channels exhibits one-shot zero-error superadditivity when forward assisted by a perfect $3$-level quantum channel, despite its corresponding confusability graph being KS colorable.
However, the set is known to exhibit state-independent contextuality \cite{Xu2015,Cabello2015}. The channel construction, encoding and decoding scheme appear in the Appendix A.
\end{remark}

In light of Remark~\ref{sic}, KS uncolorability is not the structural feature underlying the effect observed here. This observation motivates a deeper structural question: what properties of the noisy classical channel enable such an enhancement, and conversely, which channels provably cannot exhibit superadditivity under this protocol. In particular, we seek to determine whether the observed advantage is governed by specific features of the confusability graph of the noisy channel.

 To address this question, we introduce a graph-theoretic framework \cite{Cabello2010,Cabello2014,Amaral2018} based on exclusivity relations between measurement events. This framework has been extensively used in the study of Kochen--Specker contextuality and provides a natural language for distinguishing classical and quantum behaviors.

Let $\mathcal{G}$ be a finite graph with vertex set $V_\G$ and edge set $E_\G$. Each vertex represents a distinct measurement event, and an edge between two vertices indicates that the corresponding events are mutually exclusive. In a probabilistic theory, each event is assigned an occurrence probability $p(v)$. A {behavior} on $\G$ is a function $p : V_\G \to [0,1]$ satisfying the pairwise exclusivity condition
$p(v) + p(v') \leq 1 \quad \text{for all } \{v,v'\} \in E_\G$.
A behavior is called \emph{classical} or \emph{noncontextual} if there exists a probability space $(\Omega,\Sigma,\mu)$ and measurable sets $A_v \in \Sigma$ such that $A_v \cap A_{v'} = \varnothing$ whenever $\{v,v'\} \in E_\G$, and $p(v) = \mu(A_v)$ for all $v \in V_\G$. We denote the set of classical behaviors by $\mathscr{C}_\G$.
A behavior is \emph{quantum} if it admits a quantum realization: there exist a Hilbert space $\mathcal{H}$, a density operator $\rho$ on $\mathcal{H}$, and projectors $\{\Pi_v\}_{v \in V(G)}$ such that $\Pi_v \Pi_{v'} = 0$ whenever $\{v,v'\} \in E(G)$ and
$p(v) = \mathrm{Tr}(\rho \Pi_v)$
for all $v \in V_\G$. The set of quantum behaviors is denoted $\mathscr{Q}_\G$. Since any classical model can be embedded into quantum theory, one always has $\mathscr{C}_\G \subseteq \mathscr{Q}_\G$. For some graphs this inclusion is strict, reflecting genuinely quantum behaviors; for others the two sets coincide, indicating that quantum resources offer no advantage within the given exclusivity structure.
In the quantum realization above, the graph $\mathcal G$ may be viewed as the
orthogonality graph induced by the family of rank one projectors. We now use this framework to derive a criterion for superadditivity based on the confusability graph associated with a channel hypergraph of the noisy classical channel.

\begin{theorem}\label{exc}
If the confusability graph $G(N_n)$ of a noisy classical channel $N_n$ is isomorphic to an exclusivity graph $\mathcal{G}$ satisfying $\mathscr{C}_\G=\mathscr{Q}_\G$, then the communication protocol considered here does not exhibit one-shot zero-error superadditivity.
\end{theorem}

\begin{proof}

It is known that $\mathscr{C}_\G=\mathscr{Q}_\G$ if and only if $\mathcal G$ is a perfect graph \cite{GLS,Knuth1994,Rosenfeld1967,Berge1961,Chudnovsky2006,Cabello2010,Cabello2014,Amaral2018}. For perfect graphs, the inequality
$\alpha(\mathcal H)\,\omega(\mathcal H) \ge |V(\mathcal H)|$
holds for every induced subgraph $\mathcal H$ of $\mathcal G$, and in particular for $\mathcal G$ itself \cite{perfect}. Writing $n=|V(\mathcal G)|$, this implies
$\frac{n}{\alpha(\mathcal G)} \le \omega(\mathcal G)$.

If $\mathcal G$ admits an orthogonal representation \cite{onlynote} by rank-one projectors acting on a Hilbert space of dimension $d$, then each clique of $\mathcal G$ corresponds to a set of mutually orthogonal projectors. This implies that $\omega(\mathcal G)\le d$, and consequently 
$\alpha(\mathcal G)\, d \ge n$.

In our protocol, classical messages are encoded in the vertices of $G(N_n)$, so that at most $n$ messages can be transmitted with zero-error. When the condition $\alpha(\mathcal G)\, d \ge n$ holds, one-shot zero-error superadditivity is therefore ruled out here.
\end{proof}
Consequently, one-shot zero-error superadditivity is possible only when the classical behavior set is a proper subset of the quantum set, $\mathscr{C}_\G\subsetneq \mathscr{Q}_\G$, reflecting 
contextual structure in the underlying graph.
Since Kochen–Specker contextuality can only arise in quantum systems above a certain dimensional threshold, this naturally motivates following no-go corollaries.
\begin{corollary}\label{c1}
For any noisy classical channel $N_n$, the parallel use of $N_n$ with a two-level perfect quantum channel cannot exhibit zero-error superadditivity under the protocol considered here.
\end{corollary}
\begin{corollary}\label{c2}
Let $N_n$ be a noisy classical channel whose confusability graph admits an orthogonal representation in $\mathbb{C}^d$, $d$ being the minimum possible dimension. If the assisting perfect quantum channel has dimension strictly less than $d$, then no one-shot zero-error superadditivity is possible under this protocol.
\end{corollary}
The proofs of the above corollaries are provided in the Appendix B.

\textit{Discussion.-}
We have explored a new kind of superadditivity in the zero-error communication scenario that is fundamentally quantum in origin. This quantum character manifests in an asymmetric communication setting where a noisy classical channel is assisted by a perfect quantum channel in parallel. In particular, we show that the one-shot zero-error classical capacity of the combined channel exceeds the product of their individual zero-error capacities, and that this enhancement disappears when the assisting quantum channel is replaced by a perfect classical channel, confirming that the effect is genuinely quantum. We further show that the superadditivity observed in this scheme is governed by the structural properties of the noisy classical channel. Specifically, superadditivity is possible only when the confusability graph of the classical channel is isomorphic to an exclusivity graph that admits strictly  quantum behavior. A sufficient condition for superadditivity in our framework is straightforward: if the number of vertices of the channel hypergraph exceeds the product of the independence number of its confusability graph and the dimension of an orthogonal representation of the graph, then the channel can be paired with a perfect quantum channel of matching dimension to exhibit superadditivity of one-shot zero-error classical capacity. Unlike earlier superadditivity results involving combinations of noisy and perfect quantum channels \cite{Duan2009,Cubitt2011a,Cubitt2012,Shirokov2015,Duan2016}, the effect reported here arises from purely noisy classical channel assisted by a perfect quantum channel.

Enhancement of the zero-error capacity of classical channels using shared entanglement was previously explored by Cubitt \textit{et al.} \cite{Cubitt2010}. 
{Subsequent analyses by Yadavalli and Kunjwal \cite{Yadavalli2022} identified an important limitation of that protocol: capacity enhancement is possible only when the channel hypergraph arises from a \emph{disjoint} Kochen–Specker basis set. Consequently, the Cubitt \textit{et al.} scheme offers no advantage for channels whose hypergraphs do not have this structure, including those associated with the Cabello \textit{et al.} construction \cite{CABELLO1996183} or other state-independent contextuality sets that remain KS-colorable. Our protocol overcomes this limitation.}
In particular, we demonstrate superadditivity for noisy classical channels whose hypergraphs coincide with those arising from the construction of Cabello \textit{et al.}, as well as for a family of channels constructed from a state-independent contextuality set that is nevertheless KS-colorable \cite{Xu2015}. Here, we also showed that a three-level quantum channel is the minimal assisting resource capable of producing this superadditivity.

This work suggests several directions for future study. Although Kochen–Specker contextuality is necessary for an orthogonality graph to support superadditivity when interpreted as a confusability graph, it is not sufficient within our present framework. Developing quantum schemes that realize superadditivity for more general contextual graphs remains open. Furthermore, in our construction the number of perfectly transmitted messages is bounded by the size of vertex set of the confusability graph. Whether this bound can be surpassed remains an open question.

\begin{acknowledgments}
\textbf{Acknowledgements.} We are especially grateful to Manik Banik for numerous stimulating and insightful discussions that significantly shaped various aspects of this work. We also thank Guruprasad Kar, Sibasish Ghosh, Mir Alimuddin, Chandan Datta, Sahil G. Naik, Subhendu B. Ghosh, Snehasish Roy Chowdhury, Ram Krishna Patra, Anandamay Das Bhowmik, Tamal Guha, Pratik Ghosal, Chitrak Roychowdhury, Tushar for insightful discussions at various stages of this work. KA acknowledges support from the CSIR project 09/0575(19300)/2024-EMR-I. AM thanks Sumit Mukherjee for insightful discussions on contextuality. AM acknowledges the EntangleX Dialogues YouTube channel for organizing a series of quantum information seminars, in particular the ones organized around Guruprasad Kar’s talk series on “Physics of Distilling Secrecy and Entanglement”.
\end{acknowledgments}




\appendix

\section{Appendix A: Encoding and decoding scheme for Remark \ref{sic}}
\label{app:superadditivity}

In this appendix we provide an explicit example underlying Remark 1. The construction is based on a family of vectors in $\mathbb{C}^3$ whose orthogonality relations admit a natural interpretation as a classical channel confusability structure. We show how this structure leads to one shot zero error superadditivity when a noisy classical channel is used in parallel with a perfect three level quantum channel.

\paragraph{A family of vectors in $\mathbb{C}^3$.}
We begin by introducing the vector set that will generate the required orthogonality relations. Fix an integer parameter $m\ge 1$ and define $\omega = e^{2\pi i/(3m)}$. Consider the following family of vectors in $\mathbb{C}^3$:
$
(1,0,0),\; (0,1,0),\; (0,0,1),
$
$
(1,-\omega^i,0),\; (1,0,-\omega^i),\; (0,1,-\omega^i),\; (1,\omega^i,\omega^j),
$
where $i,j\in\{1,\dots,3m\}$. Let $\mathcal{V}_m$ denote the set of all such vectors. A direct count shows that
$
|\mathcal{V}_m| = 3 + 9m + 9m^2.
$

The orthogonality relations among these vectors play a central role. They induce a graph structure in which vertices correspond to elements of $\mathcal{V}_m$ and edges connect pairs of orthogonal vectors.

\paragraph{Orthogonality graph and contextuality properties.}
Let $\mathcal{G}_m$ be the orthogonality graph associated with $\mathcal{V}_m$, whose vertices are the vectors in $\mathcal{V}_m$ and where two vertices are adjacent if and only if the corresponding vectors are orthogonal. For a vertex $v\in\mathcal{V}_m$ with vector $(a,b,c)$ we associate the normalized qutrit state
$
\ket{\psi_v}=\tfrac{a\ket{0}+b\ket{1}+c\ket{2}}{\|(a,b,c)\|}\in\mathbb{C}^3,
$
and the corresponding rank one projector $\Pi_v=\ket{\psi_v}\!\bra{\psi_v}$. Adjacency of two vertices $v$ and $w$ therefore implies $\Pi_v\Pi_w=0$. The resulting family of projectors is known to constitute a state independent contextuality set in Hilbert space of dimension three, while remaining $\{0,1\}$-colorable \cite{Xu2015}.

\paragraph{Construction of a noisy classical channel.}
Here we provide an explicit construction of the family of channels $N_n^m$, whose input alphabet has cardinality $3 + 3k + 3k^2.$
It is convenient to describe the channel in terms of its channel hypergraph, where each hyperedge corresponds to a set of input symbols that may produce the same channel output. By construction, every hyperedge consists of vertices whose associated vectors are mutually orthogonal, and hence each hyperedge contains either two or three vertices, depending on the orthogonality relations in $\mathcal{V}_m$. The hypergraph has 
$3 + 3k + 3k^2$
vertices, where 
$k = 3m$, $m \in \mathbb{Z}_{+},$
and the channel outputs are identified with the hyperedges of this hypergraph. When $m$ is odd, the number of three-vertex hyperedges is 
$\frac{k^2}{3} + 1,$
and the number of two-vertex hyperedges is 
$3k^2 + 3k.$
When $m$ is even, the number of three-vertex hyperedges is 
$\frac{k^2}{3} + \frac{3k}{2} + 1,$
and the number of two-vertex hyperedges is 
$3k^2.$ 
We interpret $\mathcal{G}_m$ as the confusability graph associated with this channel hypergraph; both the confusability graph and the channel hypergraph share the same vertex set and therefore have the same number of vertices.

\paragraph{Capacity parameters.}
The one shot zero error classical capacity of $N_n^m$ is given by the independence number of its confusability graph, namely $\alpha(\mathcal{G}_m)$ which is known and given by
$
\alpha(\mathcal{G}_m) = 3m(m+1)$.
 Since the vectors live in $\mathbb{C}^3$, we denote by $d=3$ the local Hilbert space dimension of the quantum system that will be transmitted through a $3$-level perfect quantum channel for assistance. A simple comparison shows that
$
|V(\mathcal{G}_m)|> \alpha(\mathcal{G}_m)\, d.
$
This numerical separation is the key to the superadditivity phenomenon exhibited below.

\paragraph{Proof of Remark 1.}
We now describe an explicit encoding and decoding protocol for the joint use of the two channels that achieves a one shot zero error capacity strictly exceeding the product of the individual one shot zero error capacities.

\paragraph{Encoding.}
The sender wishes to transmit a message drawn from the message set $\mathcal{M}_m$ with $|\mathcal{M}_m| = |V(\mathcal{G}_m)|$. Fix a bijective encoding map $\mathscr{E}:\mathcal{M}_m\to V(\mathcal{G}_m)$. To transmit a message $u\in\mathcal{M}_m$, the sender proceeds as follows. First, the classical label corresponding to $v=\mathscr{E}(u)$ is sent through the noisy classical channel $N_n^m$. Concurrently the sender transmits the qutrit state
$
\ket{\phi_v}=\ket{\psi_v}
$
through a perfect three level quantum channel, where $\ket{\psi_v}$ is the normalized state associated with the vector that labels $v$.

\paragraph{Channel output.}
By construction of $N_n^m$, the classical channel output is a hyperedge $h\subseteq V(\mathcal{G}_m)$ such that $v\in h$. Through the quantum channel the receiver obtains the corresponding qutrit state $\ket{\phi_v}$, whose identity is a priori unknown to the receiver.

\paragraph{Decoding.}
Upon receiving the hyperedge $h$ and the quantum state $\ket{\phi_v}$, the receiver implements the following projective measurement determined by $h$. If $h$ contains three vertices, the receiver measures in the orthonormal basis formed by the three corresponding qutrit states; the three rank one projectors constitute the measurement. If $h$ contains two vertices, let $\Pi_1$ and $\Pi_2$ denote the two rank one orthogonal projectors onto the corresponding qutrit states; the measurement is completed by the unique third projector
$
\Pi_3 = \mathbb{I}_3 - (\Pi_1 + \Pi_2).
$
In either case $\ket{\phi_v}$ is one of the measurement basis elements, so the measurement outcome identifies $v$ with certainty, and hence it maps the identified vertex back $v$ to the message $u$.

\paragraph{Superadditivity.}
The protocol succeeds with zero probability of error for every message. Hence the one shot zero error classical capacity of the joint use of $N_n^m$ and the perfect three level quantum channel is at least $|V(\mathcal{G}_m)|$. The individual one shot zero error classical capacities are $\alpha(\mathcal{G}_m)$ for $N_n^m$ and $d=3$ for the perfect three level quantum channel, so the product of the individual capacities equals $3\alpha(\mathcal{G}_m)$. Using $\alpha(\mathcal{G}_m)=3m(m+1)$ and $|V(\mathcal{G}_m)|=3+9m+9m^2$ we obtain
$
|V(\mathcal{G}_m)| = 3\alpha(\mathcal{G}_m) + 3 \;>\; 3\alpha(\mathcal{G}_m)
$ for all $m\in\mathbb{Z}_+$.
Therefore the capacity achieved by the joint channel strictly exceeds the product of the individual capacities. This completes the proof of Remark 1 and provides an explicit family of noisy classical channels that exhibit the superadditivity phenomenon when used in parallel with a perfect three level quantum channel, in accordance with Theorem 1. The underlying construction of the channel confusability graph is based on a state independent contextuality graph that nevertheless admits a KS coloring.

\section{Appendix B: Proofs of Corollary 1 and 2} \label{app:corollary}

\textbf{Corollary 1.} For any noisy classical channel $N_n$, the parallel use of $N_n$ with a two-level perfect quantum channel cannot exhibit zero-error superadditivity under the protocol considered here.

\begin{proof}
    For projective measurements, any exclusivity graph that admits an orthogonal representation in dimension 2 has a noncontextual hidden variable model. Consequently, the corresponding quantum behavior set coincides with the classical behavior set. Therefore, in our protocol, the quantum states corresponding to each vertex must belong to a Hilbert space of dimension greater than 2. 

    Now, if the classical channel is assisted by a $2$-level perfect quantum channel defined by the map $\mathcal{I}_2:\mathcal{D}(\mathbb{C}^2)\mapsto\mathcal{D}(\mathbb{C}^2)$ with $\mathcal{I}_2({\rho})=\rho,$ for all $\rho\in \mathcal{D}(\mathbb{C}^2)$, then a three or higher dimensional quantum system can not be transmitted intact through this channel from the sender to the receiver. 
    This completes the proof. 
\end{proof}

\textbf{Corollary 2.} Let $N_n$ be a noisy classical channel whose confusability graph admits an orthogonal representation in $\mathbb{C}^d$, $d$ being the minimum possible dimension. If the assisting perfect quantum channel has dimension strictly less than $d$, then no one-shot zero-error superadditivity is possible under this protocol.
\begin{proof}
    The proof of this corollary is similar to that of Corollary 1. A $(d-1)$-level perfect quantum channel cannot transmit a $d$-dimensional quantum system intact. Therefore, perfect decoding of the transmitted message is not possible.
\end{proof}



\end{document}